\documentclass[a4paper,11pt]{article}
\usepackage{lmodern}
\usepackage[utf8]{inputenc}
\usepackage[T1]{fontenc}
\usepackage[USenglish]{babel}
\usepackage{amsmath,amssymb,amsthm}

\usepackage{fullpage}

\usepackage{tikz} 

\usepackage{enumitem}
\usepackage[labelsep=period,labelfont={sc},textfont=it,margin=\parindent,font=small,format=plain]{caption}
\usepackage[numbers,sort&compress]{natbib}
\usepackage[%
  pdftitle={The eternal dominating set problem for interval graphs},%
  pdfauthor={Martín Rinemberg and Francisco J.\ Soulignac},%
  pdfcreator={},%
  pdfsubject={},%
  pdfkeywords={},%
  colorlinks=true,%
  linkcolor=blue,%
  citecolor=blue,%
  urlcolor=blue]{hyperref}  
\usepackage{doi}
  
\newtheorem*{theorem*}{Theorem}

\newcommand{\I}{\mathcal{I}}
\newcommand{\V}{\mathcal{V}}
\newcommand{\Dom}{\gamma}
\newcommand{\ConDom}{\Dom_c}
\newcommand{\Cover}{\theta}
\newcommand{\NeoCol}{\Cover_c}

\newcommand{\Eternal}{\Dom^{\infty}}
\newcommand{\oEternal}{\Eternal_{1,1}}
\newcommand{\mEternal}{\Eternal_{n,1}}
\newcommand{\mmEternal}{\Eternal_{n,n}}
\newcommand{\Weight}{\omega}
\newcommand{\Ind}{\alpha}

\begin{document}

\title{The eternal dominating set problem for interval graphs}

\author{Mart\'\i n Rinemberg\thanks{Departamento de Computación, FCEN, Universidad de Buenos Aires, Buenos Aires, Argentina.} \and Francisco J.\ Soulignac\thanks{CONICET and Departamento de Ciencia y Tecnología, Universidad Nacional de Quilmes, Bernal, Argentina.}}

\date{\normalsize\texttt{mrinemberg@dc.uba.ar}, \texttt{francisco.soulignac@unq.edu.ar}}

\maketitle

\begin{abstract}
We prove that, in games in which all the guards move at the same turn, the eternal domination and the clique-connected cover numbers coincide for interval graphs.  A linear algorithm for the eternal dominating set problem is obtained as a by-product.
\end{abstract}

Consider a game, defined by parameters $x, y \in \mathbb{N}$, that is played by an attacker and a defender on a graph $G$ with $n$ vertices.  Initially, the defender places $k \leq n$ guards on the vertices of $G$, leaving at most $y$ guards on each vertex.  At each turn, the attacker first \emph{attacks} a vertex $v$ of $G$, and the defender then moves at most $x$ guards.  Each moved guard is transferred to a neighbor of its current vertex, leaving at most $y$ guards on a same vertex.  The attack is \emph{repelled} if a guard occupies $v$ after the defender's move.  The attacker wins the game if one of its attacks is not repelled, whereas the defender wins if it is able to eternally repel the attacks.  In the latter case, the initial multiset of occupied vertices is an \emph{$(x,y)$-eternal dominating set}.  The \emph{$(x,y)$-eternal domination number} $\Eternal_{x,y}(G)$ of $G$ is the minimum $k$ such that the defender wins the game.  In the following, we refer to the game defined by parameters $x,y$ as the \emph{$(x,y)$-game}.

Let $V(G)$ be the vertex set of $G$ and $G[V]$ be the subgraph of $G$ induced by $V \subseteq V(G)$.  A set $D \subseteq V(G)$ is \emph{connected} when $G[D]$ is connected and \emph{dominating} if each vertex in $V(G) \setminus D$ has a neighbor in $D$.  The (\emph{connected}) \emph{domination number} $\Dom(G)$ ($\ConDom(G)$) is the minimum $k$ such that $G$ has a (connected) dominating set with $k$ vertices.  Clearly, the attacker wins an $(n,n)$-game when the initial set of occupied vertices is not dominating.  Conversely, a non-trivial connected graph $G$ can be defended in an $(n,1)$-game by placing a ``stationary'' guard at each vertex of a connected dominating set $D$ and a ``rover'' guard at a vertex $w \not\in D$.  In each turn, the defender virtually translates the rover to the attacked vertex $v$ by moving all the guards in a shortest path from $w$ to $v$ whose interior belongs to $D$. Hence, $\Dom(G) \leq \mmEternal(G) \leq \mEternal(G) \leq 1+\ConDom(G)$.

Similar bounds for $(1,1)$-games follow by considering the independence and clique cover numbers.  A \emph{clique} (an \emph{independent set}) of $G$ is a set of pairwise adjacent (non-adjacent) vertices of $G$, while a \emph{clique cover} of $G$ is a partition of $V(G)$ into cliques.  The \emph{independence number} $\Ind(G)$ of $G$ is the maximum such that $G$ has an independent set of size $\alpha(G)$, while the \emph{clique cover number} $\Cover(G)$ of $G$ is the minimum such that $G$ has a clique cover with $\Cover(G)$ parts.  If a different vertex of an independent set is attacked at each turn, then different guards are required to repel the attacks.  Conversely, every attack can be repelled if one guard defends each clique of a clique cover.  Thus, $\alpha(G) \leq \oEternal(G) \leq \Cover(G)$.  

As noted by Goddard et al.~\cite{GoddardHedetniemiHedetniemiJCMCC2005}, both upper bounds can be strengthened in $(n,1)$-games.  Define the \emph{weight} of a connected set $V$ as $\Weight(V) = 1$ if $V$ is a clique and $\Weight(V) = 1+\ConDom(G[V])$ otherwise.  A \emph{neocolonization} of $G$ is a partition $\V$ of $V(G)$ into connected sets; its \emph{weight} is $\Weight(\V) = \sum_{V \in \V} \Weight(V)$.  The \emph{clique-connected cover number} $\NeoCol(G)$ of $G$ is the minimum such that $G$ admits a neocolonization of weight $\NeoCol(G)$.  By the previous discussion, $\mEternal(G) \leq \NeoCol(G) \leq \min\{\Cover(G), \ConDom(G)+1\}$.

Together with $\mEternal(G) \leq \Ind(G)$, the above are some of the \emph{elementary} bounds that were discovered since the eternal domination problems were introduced in~\cite{BurgerCockayneGruendlinghMynhardtVuurenWinterbachJCMCC2004} and~\cite{GoddardHedetniemiHedetniemiJCMCC2005}; see~\cite{KlostermeyerMynhardtAADM2016} for an up-to-date review.  A nice feature about these inequalities is that they are easy to prove: lower bounds follow from simple greedy attack sequences, while upper bounds are obtained by partitioning $G$ into easy-to-defend subgraphs.  None of the inequalities in the chains $\Dom(G) \leq \mmEternal(G) \leq \mEternal(G) \leq \Ind(G) \leq \oEternal(G) \leq \Cover(G)$ and $\mEternal(G) \leq \NeoCol(G) \leq \Cover(G)$ holds by equality for all graphs (see~\cite{KlostermeyerMynhardtAADM2016}).  Yet, equality holds for certain graph classes, e.g., $\Ind(G) = \oEternal(G) = \Cover(G)$ when $G$ is perfect (c.f.\ above) and $\mEternal(G) = \NeoCol(G)$ when $G$ is a tree~\cite{KlostermeyerMacGillivrayJCMCC2009}.  In a recent article, Braga et al.~\cite{BragaSouzaLeeIPL2015} show that $\mmEternal(G) = \Cover(G)$ for proper interval graphs.  In this note we generalize their result by proving that $\mmEternal(G) = \NeoCol(G)$ for all interval graphs.  Whereas Braga et al.\ derive non-trivial lower bounds of $\mmEternal(G)$ for general graphs, we obtain a short proof, similar on spirit to those of the elementary bounds, by restricting our attention to interval graphs. 

\begin{figure}
\centering
\begin{tikzpicture}[ultra thick,scale=.45,yscale=.65]
    \def\Intervals{%
         0/ 3/1/A_1=D_1/3.7,%
         2/ 5/2//,%
         6/ 7/2/A_2=D_2/5,%
         8/ 9/2/A_3/3.7,%
         4/11/1/D_3/5,%
        12/13/1/A_4/3.7,%
        10/15/2/D_4/5,%
        14/17/1//,%
        16/21/2//,%
        18/19/1/A_5=D_5/3.7,%
        22/23/2/A_6=D_6/5,%                
        24/25/2/A_7/3.7,%                
        20/27/1/D_7/5,%
        26/29/2//,%
        28/31/1/A_8=D_8/3.7%        
    }
    
    \foreach \s/\t/\l/\n/\p in \Intervals {
        \draw (\s,\l) to (\t, \l);
        \if\n\relax\else\draw [thin,gray] (\t,\l) to (\t, \p) node [above,black] {$\n$};\fi
    }
    
    \draw [thin,dashed] (-.5,0) to (31.5,0);
    \draw [thin,dashed] (-.5,3) to (31.5,3);
    \foreach \x in {3.5,15.5,19.5,27.5} {
      \draw [thin,dashed] (\x,0) to (\x, 3);        
    }
    
    \foreach \i/\p in {1/1.5,2/9.5,3/17.5,4/23.5,5/29.5} {
      \node at (\p,-1) {$\I_\i$};
    }
\end{tikzpicture}

\begin{tikzpicture}[yscale=.6,Vertex/.style={circle,thick,draw,fill=black,text=white,inner sep=0pt,minimum size=4pt}]
 \foreach \i\l in {1/v(A_1),2/,3/v(D_3),4/v(D_4),5/,6/,7/v(D_7),8/,9/v(D_8)} {
    \node [Vertex] (\i) at (1.5*\i,0) [label=above:$\l$]{};
 }
 \foreach \i in {1,...,8} {
    \pgfmathsetmacro\n{\i+1}%
    \draw [thick] (\i) to (\n); 
 }
 \foreach \i/\d\/\l in {3/-1/v(A_2),3/1/,4/0/v(A_4),6/0/v(A_5),7/-1/,7/1/v(A_7)} {
   \node [Vertex] (b\i\d) at (1.5*\i+.4*\d,-1) [label=below:$\l$] {};
   \draw [thick] (\i) to (b\i\d);
 }
\end{tikzpicture}

\caption{An interval model $\I$ (above) of an interval graph $G$ (below).  The labels for the intervals correspond to those in the proof of the Theorem.}\label{fig}
\end{figure}
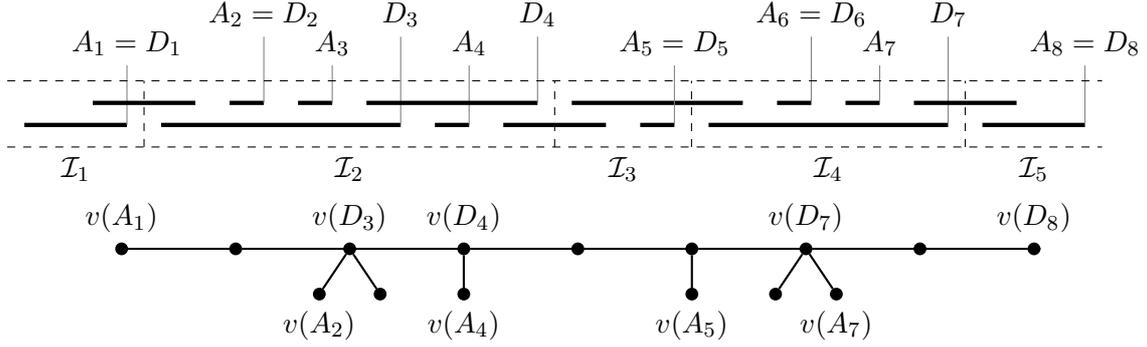

A graph $G$ is an \emph{interval graph} when each $v \in V(G)$ can be mapped into an interval $I(v)$ of the real line in such a way that $v$ and $w$ are adjacent in $G$ if and only if $I(v) \cap I(w) \neq\emptyset$.  The family $\I = \{I(v) \mid v \in V(G)\}$ is an \emph{interval model} of $G$ (Figure~\ref{fig}).  If no interval of $\I$ contains another interval of $\I$, then $\I$ is a \emph{proper interval model} and $G$ is a \emph{proper interval graph}.  Let $v(I)$ denote the vertex of $G$ corresponding to $I \in \I$, and $v(\I') = \{v(I) \mid I \in \I'\}$ for $\I' \subseteq \I$.  We write $s(I) = s$ and $t(I) = t$ to denote the \emph{beginning} and \emph{ending} points of $I = (s,t)$.    

\begin{theorem*}
 If $G$ is an interval graph, then $\mmEternal(G) = \mEternal(G) = \NeoCol(G)$.
\end{theorem*}

\begin{proof}
 Without loss of generality, $G$ has an interval model $\I$ whose intervals all have different endpoints (Figure~\ref{fig}).  Define $A_1, \ldots, A_k$ and $D_1, \ldots, D_k$ as the maximal sequences such that $A_1 = D_1$ is the interval of $\I$ with minimum ending point, and:
 \begin{enumerate}[label=(\alph*)]
  \item $B_{i+1}$ is the interval with ending point $\max\{t(I) \mid I \in \I \text{ and } s(I) < t(D_i)\}$,\label{B}
  \item $A_{i+1}$ is the interval with ending point $\min\{t(I) \mid I \in \I \text{ and } s(I) > t(D_i)\}$, and\label{A}
  \item $D_{i+1} = A_{i+1}$ if $t(A_{i+1}) > t(B_{i+1})$ and $D_{i+1} = B_{i+1}$ otherwise.\label{D}
 \end{enumerate}
 Applying \ref{B}--\ref{D} by induction, it is easy to see that $s(I(v)) < t(D_{i+j-1})$ when $v \in V(G)$ is at distance at most $j$ from $v(A_i)$, for $i < i+j \leq k$.  Therefore, by~\ref{A}, $v(A_{i+j})$ is at distance at least $j+1$ from $v(A_i)$.  Thus, the guards that occupy $v(A_i)$ at turn $i$ cannot occupy $v(A_{i+j})$ at turn $i+j$.  Hence, at least $k$ guards are required to repel the first $k$ attacks when $v(A_i)$ is attacked at turn $i$.  That is, $\mmEternal(G) \geq k$.
 
 We now prove that $k \geq \NeoCol(G)$.  For the sake of notation, let $A_0 = D_0$ and $A_{k+1} = D_{k+1}$ be intervals outside $\I$ with $t(A_0) < \min\{s(I) \mid I \in \I\}$ and $s(A_{k+1}) > \max\{t(I) \mid I \in \I\}$.  Thus, $p_0 = 0$, $p_1 = 1$, and $p_h = k+1$ for the indices $p_0 < \ldots < p_h$ such that $D_{p_i} = A_{p_i}$.  Fix $1 \leq i < h$, and let $p = p_i$, $q = p_{i+1}$, and $\I_i = \{I \in \I \mid t(D_{p-1}) < s(I) < t(D_{q-1})\}$ (Figure~\ref{fig}).  If $q = p+1$, then, by~\ref{A}, $t(A_p) = t(D_{q-1})$ is the lowest ending point in $\I_i$, thus $v(\I_i)$ is a clique (of weight $p_{i+1}-p_i$).  Otherwise, by \ref{B}~and~\ref{D}, $D = D_{p+1} \cup \ldots \cup D_{q-1}$ is the interval $(s(D_{p+1}), t(D_{q-1}))$.  Moreover, by~\ref{A}, $t(A_p) = t(D_p) > s(D)$ is the lowest ending point in $\I_i$, thus every interval in $\I_i$ has least one endpoint inside $D$.  In other words, $v(D_{p+1}), \ldots, v(D_{q-1})$ is a connected dominating set of $G[v(\I_i)]$ and, therefore, $\Weight(v(\I_i)) \leq q-p = p_{i+1} - p_i$.  Summing up, $\V = \{v(\I_1), \ldots, v(\I_{h-1})\}$ is a neocolonization of $G$ with weight 
 \begin{displaymath}
  \Weight(\V) = \sum_{i=1}^{h-1} \Weight(v(\I_i)) \leq \sum_{i=1}^{h-1}(p_{i+1}-p_i) = p_h - p_1 = k.
 \end{displaymath}
 Hence, $\NeoCol(G) \leq k \leq \mmEternal(G) \leq \mEternal(G) \leq \NeoCol(G)$ as desired.
\end{proof}

We remark that if $\I$ is proper, then $A_i = D_i$ for $1 \leq i \leq k$.  Hence, $A_1, \ldots, A_k$ is an independent set of $G$ and $v(\I_1), \ldots, v(\I_k)$ is a clique cover of $G$. So, $\Cover(G) \leq \mmEternal(G) \leq \Ind(G) \leq \Cover(G)$.  That is, the above proof implies the result by Braga et al.  In the general case, $\V = \{v(\I_1), \ldots, v(\I_{h-1})\}$ is a neocolonization of minimum weight.  Moreover, $v(\I_i)$ can be eternally defended if one guard is initially placed at $v(D_j)$ for $p \leq j < q$, because $v(D_{p+1}), \ldots, v(D_{q-1})$ is a connected dominating set of $G[v(\I_i)]$ when $v(\I_i)$ is not a clique.  Therefore, $D = \{v(D_i) \mid 1 \leq i \leq k\}$ is an $(n,y)$-eternal dominating set of $G$ of minimum size, for $y \in \{1,n\}$.  

To compute $\V$ and $D$ it suffices to find the maximal sequences $A_1, \ldots, A_k$ and $D_1, \ldots, D_k$ of\/ $\I$ satisfying \ref{B}--\ref{D}.  Suppose $D_i$ was found at step $i$, $0 \leq i \leq k$, where $t(D_0) < \min\{t(I) \mid I \in \I\}$.  For step $i+1$, the ending points after $t(D_i)$ are traversed to find the interval $A_{i+1}$ satisfying~\ref{A}.  If $A_{i+1}$ does not exist, then the algorithm ends as $k=i$.  Otherwise, the beginning points in $D_i$ are traversed to find the interval $B_{i+1}$ satisfying~\ref{B}; $D_{i+1}$ is then computed according to~\ref{D}.  As every endpoint of $\I$ is traversed $O(1)$ times, the algorithm costs $O(n)$ time when the endpoints are integers in $(0,2n]$ and $t(D_0) = 0$.  Such an interval model $\I$ can be computed in linear time from $G$ (e.g.~\cite{BoothLuekerJCSS1976}), thus $\V$ and $D$ can be found in linear time when the input is either $G$ or $\I$.  

\small
%\bibliographystyle{notabbrvnat}
%\bibliography{biblio}

\begin{thebibliography}{6}
\providecommand{\natexlab}[1]{#1}
\providecommand{\url}[1]{\texttt{#1}}
\expandafter\ifx\csname urlstyle\endcsname\relax
  \providecommand{\doi}[1]{doi: #1}\else
  \providecommand{\doi}{doi: \begingroup \urlstyle{rm}\Url}\fi

\bibitem[Booth and Lueker(1976)]{BoothLuekerJCSS1976}
K.~S. Booth and G.~S. Lueker.
\newblock Testing for the consecutive ones property, interval graphs, and graph
  planarity using {$PQ$}-tree algorithms.
\newblock \emph{J. Comput. System Sci.}, 13\penalty0 (3):\penalty0 335--379,
  1976.
\newblock \doi{10.1016/S0022-0000(76)80045-1}.

\bibitem[Braga et~al.(2015)Braga, de~Souza, and Lee]{BragaSouzaLeeIPL2015}
A.~Braga, C.~C. de~Souza, and O.~Lee.
\newblock The eternal dominating set problem for proper interval graphs.
\newblock \emph{Inform. Process. Lett.}, 115\penalty0 (6-8):\penalty0 582--587,
  2015.
\newblock \doi{10.1016/j.ipl.2015.02.004}.

\bibitem[Burger et~al.(2004)Burger, Cockayne, Gr\"undlingh, Mynhardt, van
  Vuuren, and
  Winterbach]{BurgerCockayneGruendlinghMynhardtVuurenWinterbachJCMCC2004}
A.~P. Burger, E.~J. Cockayne, W.~R. Gr\"undlingh, C.~M. Mynhardt, J.~H. van
  Vuuren, and W.~Winterbach.
\newblock Infinite order domination in graphs.
\newblock \emph{J. Combin. Math. Combin. Comput.}, 50:\penalty0 179--194, 2004.

\bibitem[Goddard et~al.(2005)Goddard, Hedetniemi, and
  Hedetniemi]{GoddardHedetniemiHedetniemiJCMCC2005}
W.~Goddard, S.~M. Hedetniemi, and S.~T. Hedetniemi.
\newblock Eternal security in graphs.
\newblock \emph{J. Combin. Math. Combin. Comput.}, 52:\penalty0 169--180, 2005.

\bibitem[Klostermeyer and
  MacGillivray(2009)]{KlostermeyerMacGillivrayJCMCC2009}
W.~F. Klostermeyer and G.~MacGillivray.
\newblock Eternal dominating sets in graphs.
\newblock \emph{J. Combin. Math. Combin. Comput.}, 68:\penalty0 97--111, 2009.

\bibitem[Klostermeyer and Mynhardt(2016)]{KlostermeyerMynhardtAADM2016}
W.~F. Klostermeyer and C.~M. Mynhardt.
\newblock Protecting a graph with mobile guards.
\newblock \emph{Appl. Anal. Discrete Math.}, 10\penalty0 (1):\penalty0 1--29,
  2016.
\newblock \doi{10.2298/AADM151109021K}.

\end{thebibliography}

\end{document}